\documentclass[12pt]{article}

\usepackage{graphicx}
\usepackage{caption}
\usepackage{subcaption}
\usepackage[table,xcdraw]{xcolor}

\usepackage{multirow}
\usepackage{booktabs}
\usepackage{cite}
\usepackage[top=0.7in, left=0.7in, right=0.7in, bottom=0.7in]{geometry}
\usepackage{setspace}
\usepackage{float}
\usepackage{amsmath}
\usepackage[bookmarks=false]{hyperref}
\hypersetup{colorlinks,allcolors=black}

\usepackage{titling}
\usepackage{scalerel,stackengine}
\stackMath

\usepackage[vlined,ruled,commentsnumbered,linesnumbered]{algorithm2e}

\SetCommentSty{mycommfont}
\SetKwProg{Fn}{Function}{}{end}
\SetKw{Continue}{continue}
\SetKw{In}{in}
\SetKw{And}{and}
\SetKw{To}{to}

\DeclareMathOperator{\E}{E}
\DeclareMathOperator{\Var}{Var}

\usepackage{graphicx}
\usepackage{amsthm}

\usepackage{cmbright} \usepackage[T1]{fontenc}
\usepackage[cm]{sfmath}
\usepackage{dsfont}

\newtheorem{theorem}{Theorem}
\newtheorem{lemma}[theorem]{Lemma}

\newtheorem{corollary}[theorem]{Corollary}

\newtheorem{definition}{Definition}
\theoremstyle{definition}

\begin{document}

\noindent
\begin{minipage}[H]{0.65\textwidth}

\large{\textbf{Identifying Correlation in Stream of Samples}}\\
\vspace{0.2cm}

\noindent \normalsize{Gu Zhenhao\footnote{School of Computing, National University of Singapore. Email: \url{guzh@nus.edu.sg}}, Zhang Hao }
\end{minipage}
\hfill
\begin{minipage}[H]{0.25\textwidth}

\includegraphics[width=\textwidth]{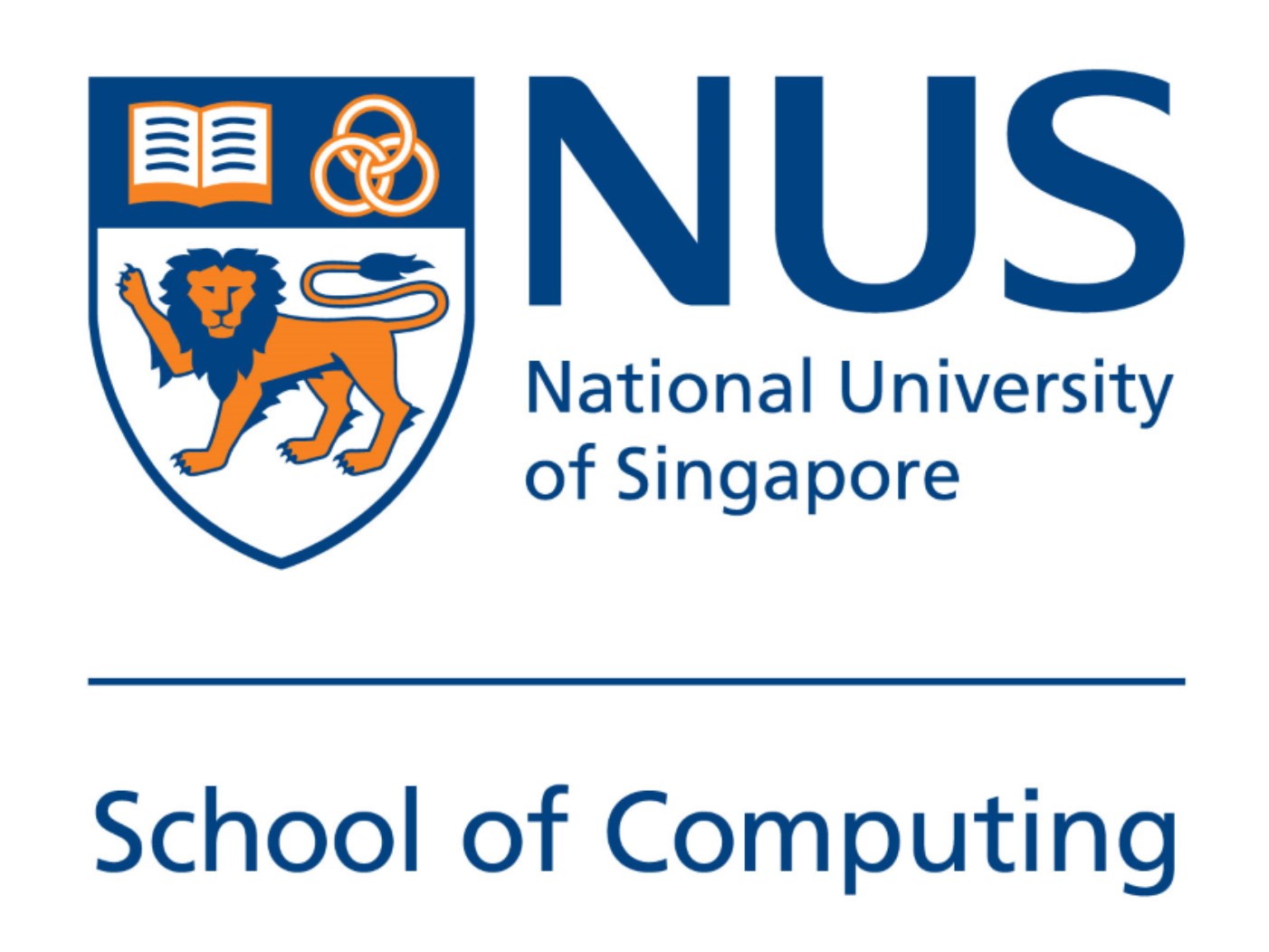}

\end{minipage}

\vspace{0.5cm}

\begin{abstract}

    \noindent Identifying independence between two random variables or correlated given their samples has been a fundamental problem in Statistics. However, how to do so in a space-efficient way if the number of states is large is not quite well-studied.
    
    \noindent We propose a new, simple counter matrix algorithm, which utilize hash functions and a compressed counter matrix to give an unbiased estimate of the $\ell_2$ independence metric. With $\mathcal{O}(\epsilon^{-4}\log\delta^{-1})$ (very loose bound) space, we can guarantee $1\pm\epsilon$ multiplicative error with probability at least $1-\delta$. We also provide a comparison of our algorithm with the state-of-the-art sketching of sketches algorithm and show that our algorithm is effective, and actually faster and at least 2 times more space-efficient.
\end{abstract}

\section{Introduction}
In this paper, we try to check whether two (or more) random variables are independent based on their samples in a data stream. More specifically, given a stream of pairs $(i,j)$, where $i,j\in[n]^2$, and $(i,j)$ follows a joint distribution $(X,Y)$, we want to check if $X$ and $Y$ are independent, i.e. $\Pr[X=i,Y=j]=\Pr[X=i]\Pr[Y=j]$ for all $i,j\in[n]^2$.

This is a fundamental problem with a lot of applications in network monitoring, sensor networks, economic/stock market analysis, and communication applications, etc.

For simplicity, we define matrix $r$ and $s$ where $r_{ij}=\Pr[X=i,Y=j]$ and $s_{ij}=\Pr[X=i]\Pr[Y=j]$. Now, to measure how close $X$ and $Y$ are to independence, we have three metrics\cite{inproceedings}:
\begin{enumerate}
    \item \textbf{$\ell_1$ difference}: The sum of absolute difference of the joint distribution and the product distribution, similar to the $\ell_1$-norm of vectors,
    $$|\Delta|=|r-s|=\sum_{i,j} |r_{ij}-s_{ij}|$$
    \item \textbf{$\ell_2$ difference}: The sum of squared difference of the joint distribution and the product distribution, taken square root, similar to the $\ell_2$-norm of vectors,
    $$\lVert \Delta\rVert=\lVert r-s\rVert=\sqrt{\sum_{i,j} (r_{ij}-s_{ij})^2}$$
    \item \textbf{Mutual information}: defined by $I(X;Y)=H(X)-H(X\mid Y)$, where $H(X)=-\sum_{i=1}^n \Pr[X=1]\log \Pr[X=i]$ is the \textit{entropy} of distribution $X$ and
    $$H(X\mid Y)=\sum_{i,j}\Pr[X=i,Y=j]\log \frac{\Pr[Y=j]}{\Pr[X=i,Y=j]}$$
    is the \textit{conditional entropy}.
\end{enumerate}
All metrics should be zero if $X$ and $Y$ are independent, and further away from 0 if $X$ and $Y$ are further away from independence (higher correlation). Our goal then reduces to give an estimation of the value of these metrics, and check if it is small enough to conclude independence between $X$ and $Y$.

Our work in this paper involves the following:
\begin{enumerate}
    \item Implement the state-of-the-art sketching of sketches algorithm proposed by Indyk and McGregor \cite{inproceedings}.
    \item Propose a novel counter matrix algorithm that store the frequency of each pair $(i,j)$ in a compressed counter matrix and calculate an unbiased estimate for the $\ell_2$ difference, and
    \item implement the counter matrix algorithm, and compare the result with the sketching of sketches algorithm.
\end{enumerate}

To avoid confusion, we summarize the notation used in the report as follows:
\begin{table}[H]
    \centering
    \begin{tabular}{ll}
    \toprule
     Notation    & Meaning \\
    \midrule
    $X,Y$ & Two random variables whose correlation we are trying to identify\\
    $n$ & Range of $X$ and $Y$ should be $[1,n]$.\\
     $p_i$ & Probability of $X$ taking value $i$, $\Pr[X=i]$\\
     $q_j$ & Probability of $Y$ taking value $j$, $\Pr[Y=j]$\\
     $s_{i,j}$  & Probability of $(i,j)$ appearing in the stream, $\Pr[X=i,Y=j]$\\
     $r_{i,j}$ & Product of marginal probabilities, $\Pr[X=i]\Pr[Y=j]$\\
     $\Delta_{i,j}$ & Difference between joint probability and product probability, $s_{i,j}-r_{i,j}$\\
     $A$ & Size of counter matrix\\
     $B$ & Number of counter matrices to run\\
     $\hat{p},\hat{q},\hat{s},\hat{\Delta}$& Estimator of $p,q,s,\Delta$ calculated by the full $n\times n$ counter matrix\\
     $\Tilde{p},\Tilde{q},\Tilde{s},\Tilde{\Delta}$ &Estimator of $\hat{p},\hat{q},\hat{s},\hat{\Delta}$ calculated by the truncated $A\times A$ counter matrix\\
     $n_{i,j}$ & The number of time each pair $(i,j)$ appears in the stream\\
     $C_{i,j}$ & The value in the counter matrix on the $i^{\text{th}}$ row and $j^{\text{th}}$ column\\
     $\sum C_{i,\cdot},\sum C_{\cdot,j}$ & The sum of values in the $i^{\text{th}}$ row, $j^{\text{th}}$ column in $C$, respectively\\
     $N$ & Length of the stream\\
    \bottomrule
    \end{tabular}
    \caption{Notations}
    \label{tab:my_label}
\end{table}
\section{Related Work}

Indyk and McGregor were the first to try to identify the correlation between sample in data stream using sketching algorithms \cite{inproceedings}. They propose an estimate for the above matrices,

\subsection{$\ell_2$ Difference}
The paper mainly utilizes the result that the weighted sum of 4-wise independent vectors $x,y\in\{-1,1\}^{n}$ and any $n\times 2$ matrix $v$ could be used to estimate the norm of $v$ \cite{ALON1999137}. More specifically,
    
    \begin{lemma}[Unbiased estimator]\label{lemma:l2}
    Consider $x,y\in\{-1,1\}^{n}$ where each vector is 4-wise independent. Let $v\in\mathds{R}^{n^2}$. Define $\Upsilon=(\sum_{i,j\in[n]}x_{i}y_{j}v_{i,j})^2$. Then $\E[\Upsilon]=\sum_{i,j\in[n]}v_{i,j}^2=\lVert v\rVert^2$ and $\Var[\Upsilon]\leq 3(\E[\Upsilon])^2$.
    \end{lemma}

    Using this lemma, we can randomly generate vectors of $x,y\in\{-1,1\}^{n}$, and let $v$ in the lemma be $r-s$, and calculate $\Upsilon$ base on the data in the stream. $\Upsilon$ is then an unbiased estimator for the $\ell_2$ difference, and the variance is also bounded, meaning that we can use Chebychev's inequality to limit the error.

 \begin{algorithm}[H]
 \vspace{0.2cm}
	\SetKwInOut{Input}{Input}
	\SetKwInOut{Output}{Output}
	\caption{Approximating $\lVert s-r\rVert^2$} \label{Alg:l2}
	$x,y\gets$ random vector of length $n$ containing $1$ or $-1$;\\
	$t_1,t_2,t_3\gets 0$;\\
	\For{each pair $(i,j)$ in stream}{
	    $t_1\gets t_1+x_iy_j$; \tcp{sums to $\sum_{i,j}x_iy_j(N\cdot s_{i,j})$}
	    $t_2\gets t_2+x_i$;  \tcp{sums to $\sum_{i}x_i(N\cdot p_{i})$}
	    $t_3\gets t_3+y_j$; \tcp{sums to $\sum_{j}y_j(N\cdot q_{j})$}
	}
	$t\gets t_1/N-t_2t_3/N^2$; \tcp{$=\sum_{i,j}x_iy_j(s_{i,j}-p_iq_j)=\sum_{i,j}x_iy_j(s_{i,j}-r_{i,j})$}
	$\Upsilon\gets t^2$; \\
	\Return $\Upsilon$;
\vspace{0.2cm}	
\end{algorithm}

The algorithm is very neat, as we only need to compute three numbers for each run. However, when it comes to implementation, the drawback of the algorithm appears.
\begin{enumerate}
    \item It is actually memory-inefficient since for each run, we need to generate random 4-wise independent vectors $x,y$ and store them until we read the whole stream. Therefore storing those vectors will actually cost $\mathcal{O}(\epsilon^{-2}n\log \delta^{-1})$ space.
    \item We can solve the above problem by choosing a large $p$ and using a polynomial of degree 3,
    $$h(x)=((h_3 x^3+h_2 x^2+h_1 x+h_0)\mod p \mod 2)\times 2-1$$
    where the coefficients $h_0,h_1,h_2,h_3<p$ are randomly chosen, to map $x$ to either $-1$ or $1$, creating a 4-wise independent vector \cite{WEGMAN1981265}. 
    
    T However, calculating this hash function is computationally expensive, making the algorithm slow. And although we get rid of the $n$ factor in space complexity, the constant factor is at least 8, so again a lot of space is used.
\end{enumerate}
We were able to verify the correctness of the algorithm, as shown in our results section.
\subsection{$\ell_1$ Difference}
The algorithm of estimating $\ell_1$ difference will take a linear space since we will need to generate and store 
\begin{itemize}
    \item A vector $x$ of length $n$ where each element follows a Cauchy distribution, and
    \item a vector $y$ of length $n$ where each element follows a $T$-truncated-Cauchy distribution \cite{indyk2006stable}
\end{itemize}

\begin{definition}[T-truncated Cauchy]\label{definition:l1}
Let $T\ge 0$, $X \sim \text{Cauchy}$,
$$
  Y = \left\{
    \begin{array}{ll}
       -T &  \text{if } X \le -T \\
      X & \text{if }  -T< X < T \\
      T & \text{if } X \ge T 
    \end{array}\right.
$$
We say $Y \sim T$-truncated-Cauchy.
\end{definition}

\begin{lemma}[Property of T-truncated Cauchy] \label{lemma:l1_1}
Let $ T = 100n $, and $ y \in \mathds{R}^n  \sim T$-Truncated-Cauchy,
\\

$$\Pr \left[ 1/100 \le \frac{\sum_j | y \cdot v^j |}{\sum_j |v^j|} \le 20 \ln n \right] \le 9/10$$
\end{lemma}

\begin{lemma}[Property of $T$-truncated Cauchy \cite{indyk2006stable}] \label{lemma:l1_2}
For any $T>0$, and $ y \in \mathds{R}^n  \sim T$-Truncated-Cauchy

$$\E \left[ \sum_j {y_j}^T \right] \le k \left[ \frac{\ln{T^2 +1}}{\pi} + b \right]$$
\end{lemma}

for each $\mathcal{O}(\log\delta^{-1})$ runs of experiment. This time, we cannot use a hash function anymore, so the $n$ factor would be inevitable for a single-pass algorithm, making it very memory-inefficient. The pseudo-code for the algorithm is as follows:

\begin{algorithm}[H]
    \label{algo_2}
	\SetKwInOut{Input}{Input}
	\SetKwInOut{Output}{Output}
	\caption{Approximating $|s-r|$} \label{Alg:l1}
	\BlankLine
	\For{$k$ in  $1...\log(1/\delta)$ in parallel}{
	    ${t_1}^1...{t_1}^s, {t_2}^1,...{t_2}^s, t_3  \gets 0$;\\
	    $T \gets 100n$;\\
	    $x_1 ... x_s \gets $ For each $ x_i^j \sim $ Cauchy and independent\\ 
	    $y_i \gets y_i \sim  T$-Truncated-Cauchy and independent \\
	    \For{each stream element (i,j)}{
	        ${t_1}^r \gets {t_1}^r + {x_i}^r \cdot y_j$; \tcp{sums to $N\sum_{i,j \in [n]} {x_i}^r y_j {r_i}_j$}
	        ${t_2}^r \gets {t_2}^r + {x_i}^r$; \tcp{sums to $N\sum_{i \in [n]} {x_i}^r p_i$}
	        ${t_3} \gets t_3 + y_j$; \tcp{sums to $N\sum_{i \in [n]} y_j q_j$}
	    }
	    $\Upsilon_k = \texttt{median}(|{t_1}^1/N - {t_2}^1\cdot t_3/N^2|, ....,|{t_1}^n/N -{t_2}^n\cdot t_3/N^2|)$;
	}
	\Return $texttt{median}(\Upsilon_1,....\Upsilon_{\log{1/\delta}})$;
\end{algorithm}

The paper proved that the algorithm is theoretically correct using definition \ref{definition:l1}, lemma \ref{lemma:l1_1}, lemma \ref{lemma:l1_2}, Markov's inequality and Chernoff Bound. This algorithm is extremely space efficient as it only takes $\mathcal{O}(\log1/\delta)$ space, however the error produced by the algorithm is high ($\mathcal{O}(\log n)$ multiplicative error) and it's dependent on $n$. Also, this algorithm assumed random oracle and infinite precision.

We tried to verify this algorithm with a few experiments, however the result is way too far from the theoretical correctness. We suspect the issue is caused by many "ideal" assumptions made in the proof and analysis, therefore the algorithm performs poorly in real life. In addition, even the expected error from its $(\mathcal{O}(\log n),\delta)$-approximation is very large, we do not see any value of further evaluating the algorithm.

\section{Preliminaries}
Given two discrete random variables $X$ and $Y$, we say that they are independent if and only if for all $i, j$, $$\underbrace{\Pr[X=i,Y=j]}_{=s_{i,j}}=\underbrace{\Pr[X=i]}_{=p_i}\underbrace{\Pr[Y=j]}_{=q_j}$$
or $s_{i,j}=p_iq_j$. This means that ideally, $\Delta_{i,j}=s_{i,j}-p_iq_j=0$ for all $i,j$ and the $\ell_1$ difference $|\Delta|=0$, $\ell_2$ difference $\lVert \Delta\rVert=0$. If we have enough space, we can store an $n\times n$ counter matrix,
\begin{table}[H]
\centering
\begin{tabular}{rccccccc
>{\columncolor[HTML]{EFEFEF}}l }
                               & 1                        & 2                        & $\dots$                                       & $j$                                                & $\dots$                  & $n-1$                    & $n$                      & $\sum$                      \\ \cline{2-8}
\multicolumn{1}{r|}{1}         & \multicolumn{1}{c|}{}    & \multicolumn{1}{c|}{}    & \multicolumn{1}{c|}{}                         & \multicolumn{1}{c|}{}                              & \multicolumn{1}{c|}{}    & \multicolumn{1}{c|}{}    & \multicolumn{1}{c|}{}    &                             \\ \cline{2-8}
\multicolumn{1}{r|}{2}         & \multicolumn{1}{c|}{}    & \multicolumn{1}{c|}{}    & \multicolumn{1}{c|}{}                         & \multicolumn{1}{c|}{}                              & \multicolumn{1}{c|}{}    & \multicolumn{1}{c|}{}    & \multicolumn{1}{c|}{}    &                             \\ \cline{2-8}
\multicolumn{1}{r|}{$\vdots$}  & \multicolumn{1}{c|}{}    & \multicolumn{1}{c|}{}    & \multicolumn{1}{c|}{}                         & \multicolumn{1}{c|}{}                              & \multicolumn{1}{c|}{}    & \multicolumn{1}{c|}{}    & \multicolumn{1}{c|}{}    &                             \\ \cline{2-9} 
\multicolumn{1}{r|}{$i$}       & \multicolumn{1}{c|}{}    & \multicolumn{1}{c|}{}    & \multicolumn{1}{c|}{}                         & \multicolumn{1}{c|}{$n_{i,j}$}                     & \multicolumn{1}{c|}{}    & \multicolumn{1}{c|}{}    & \multicolumn{1}{c|}{}    & $\sum n_{i,\cdot}$                       \\ \cline{2-9} 
\multicolumn{1}{r|}{$\vdots$}  & \multicolumn{1}{c|}{}    & \multicolumn{1}{c|}{}    & \multicolumn{1}{c|}{}                         & \multicolumn{1}{c|}{}                              & \multicolumn{1}{c|}{}    & \multicolumn{1}{c|}{}    & \multicolumn{1}{c|}{}    &                             \\ \cline{2-8}
\multicolumn{1}{r|}{$n-1$}     & \multicolumn{1}{c|}{}    & \multicolumn{1}{c|}{}    & \multicolumn{1}{c|}{}                         & \multicolumn{1}{c|}{}                              & \multicolumn{1}{c|}{}    & \multicolumn{1}{c|}{}    & \multicolumn{1}{c|}{}    &                             \\ \cline{2-8}
\multicolumn{1}{r|}{$n$}       & \multicolumn{1}{c|}{}    & \multicolumn{1}{c|}{}    & \multicolumn{1}{c|}{}                         & \multicolumn{1}{c|}{}                              & \multicolumn{1}{c|}{}    & \multicolumn{1}{c|}{}    & \multicolumn{1}{c|}{}    &                             \\ \cline{2-8}
\cellcolor[HTML]{EFEFEF}$\sum$ & \cellcolor[HTML]{EFEFEF} & \cellcolor[HTML]{EFEFEF} & \multicolumn{1}{c|}{\cellcolor[HTML]{EFEFEF}} & \multicolumn{1}{c|}{\cellcolor[HTML]{EFEFEF}$\sum n_{\cdot,j}$} & \cellcolor[HTML]{EFEFEF} & \cellcolor[HTML]{EFEFEF} & \cellcolor[HTML]{EFEFEF} & \cellcolor[HTML]{C0C0C0}$N$
\end{tabular}
\caption{$n\times n$ counter matrix}
\end{table}
where each counter $n_{i,j}$ stores the number of times the pair $(i,j)$ appears in the stream.  We can estimate the joint probability $\hat{s}_{i,j}= n_{i,j}/N$. The sum of each row and each column are calculated to estimate the marginal probabilities,
$$\hat{p}_i= \frac{\sum{n_{i,\cdot}}}{N},\qquad \hat{q}_j=\frac{\sum{n_{\cdot,j}}}{N}$$
Ideally, if $X$ and $Y$ are independent, we should have
$$\hat{\Delta}_{i,j}:=\hat{s}_{i,j}-\hat{p}_i\hat{q}_j=0$$
for all $i,j$. The $\ell_1$ difference can be estimated by
\begin{equation}\label{eq:l1}
    \widehat{|s-r|}=|\hat{\Delta}|=\sum_{i,j}|\hat{s}_{i,j}-\hat{p}_i\hat{q}_j|
\end{equation}
and the $\ell_2$ difference can be estimated by
\begin{equation}\label{eq:l2}
\widehat{\lVert s-r\rVert}=\lVert\hat{\Delta}\rVert=\sqrt{\sum_{i,j}(\hat{s}_{i,j}-\hat{p}_i\hat{q}_j)^2}
\end{equation}
We can also calculate the $\chi^2$ statistics for independence test,
\begin{equation}\label{eq:chisq}
    X_{(n-1)^2}^2=\sum_{i,j}\frac{(\hat{s}_{i,j}-\hat{p}_i\hat{q}_j)^2}{\hat{p}_i\hat{q}_j}
\end{equation}
follows a chi-squared distribution with degree of freedom $(n-1)^2$. This statistics can be compared with the critical value $\chi_{0.05,(n-1)^2}^2$ to see if we can reject a null hypothesis: $X$ and $Y$ are independent with $p<0.05$ \cite{zibran2007chi}.
\section{Counter Matrix Algorithm}
The counter matrix above is nice, but it will cost $\mathcal{O}(n^2)$ space, which is bad. The space complexity would be even worse if we want to identify correlation between multiple random variables. Our idea is \textbf{to use an $A\times A$ $(A\ll n)$ matrix to store the information in $n\times n$ counter matrix by randomly summing rows and columns together}.

More specifically, we will choose two random hash functions $h_1,h_2:[n]\rightarrow [A]$, $x\mapsto i$ where $1\leq i\leq A$. Upon seeing a pair $(x,y)$ in the stream, we increment the counter $C_{h_1(x),h_2(y)}$.
\begin{table}[H]
\centering
\begin{tabular}{rlllllllll}
\multicolumn{1}{l}{}          & \multicolumn{1}{c}{1}                         & \multicolumn{1}{c}{2}                         & \multicolumn{1}{c}{$\dots$}                   & \multicolumn{1}{c}{$j$}                                & \multicolumn{1}{c}{$\dots$}                   & \multicolumn{1}{c}{$A-1$}                     & \multicolumn{1}{c}{$A$}                       &    \multicolumn{1}{c}{$\sum$}                &                             \\ \cline{2-8}
\multicolumn{1}{r|}{1}        & \multicolumn{1}{l|}{}                         & \multicolumn{1}{l|}{}                         & \multicolumn{1}{l|}{}                         & \multicolumn{1}{l|}{\cellcolor[HTML]{FFFC9E}}          & \multicolumn{1}{l|}{}                         & \multicolumn{1}{l|}{}                         & \multicolumn{1}{l|}{}                         &                                               \\ \cline{2-8}
\multicolumn{1}{r|}{2}        & \multicolumn{1}{l|}{}                         & \multicolumn{1}{l|}{}                         & \multicolumn{1}{l|}{}                         & \multicolumn{1}{l|}{\cellcolor[HTML]{FFFC9E}}          & \multicolumn{1}{l|}{}                         & \multicolumn{1}{l|}{}                         & \multicolumn{1}{l|}{}                         &                                               \\ \cline{2-8}
\multicolumn{1}{r|}{$\vdots$} & \multicolumn{1}{l|}{}                         & \multicolumn{1}{l|}{}                         & \multicolumn{1}{l|}{}                         & \multicolumn{1}{l|}{\cellcolor[HTML]{FFFC9E}}          & \multicolumn{1}{l|}{}                         & \multicolumn{1}{l|}{}                         & \multicolumn{1}{l|}{}                         &                                               \\ \cline{2-8}
\multicolumn{1}{r|}{$i$}      & \multicolumn{1}{l|}{\cellcolor[HTML]{96FFFB}} & \multicolumn{1}{l|}{\cellcolor[HTML]{96FFFB}} & \multicolumn{1}{l|}{\cellcolor[HTML]{96FFFB}} & \multicolumn{1}{c|}{\cellcolor[HTML]{FFFC9E}$C_{i,j}$} & \multicolumn{1}{l|}{\cellcolor[HTML]{96FFFB}} & \multicolumn{1}{l|}{\cellcolor[HTML]{96FFFB}} & \multicolumn{1}{l|}{\cellcolor[HTML]{96FFFB}} & \cellcolor[HTML]{68CBD0}$\sum C_{i,\cdot}$ \\ \cline{2-8}
\multicolumn{1}{r|}{$\vdots$} & \multicolumn{1}{l|}{}                         & \multicolumn{1}{l|}{}                         & \multicolumn{1}{l|}{}                         & \multicolumn{1}{l|}{\cellcolor[HTML]{FFFC9E}}          & \multicolumn{1}{l|}{}                         & \multicolumn{1}{l|}{}                         & \multicolumn{1}{l|}{}                         &                                               \\ \cline{2-8}
\multicolumn{1}{r|}{$A-1$}    & \multicolumn{1}{l|}{}                         & \multicolumn{1}{l|}{}                         & \multicolumn{1}{l|}{}                         & \multicolumn{1}{l|}{\cellcolor[HTML]{FFFC9E}}          & \multicolumn{1}{l|}{}                         & \multicolumn{1}{l|}{}                         & \multicolumn{1}{l|}{}                         &                                               \\ \cline{2-8}
\multicolumn{1}{r|}{$A$}      & \multicolumn{1}{l|}{}                         & \multicolumn{1}{l|}{}                         & \multicolumn{1}{l|}{}                         & \multicolumn{1}{l|}{\cellcolor[HTML]{FFFC9E}}          & \multicolumn{1}{l|}{}                         & \multicolumn{1}{l|}{}                         & \multicolumn{1}{l|}{}                         &                                               \\ \cline{2-8}
\multicolumn{1}{r}{$\sum$}          &                                               &                                               &                                               & \cellcolor[HTML]{FFC702}$\sum C_{\cdot,j}$          &                                               &                                               &                                               &                 $N$                             

\end{tabular}
\caption{$A\times A$ Counter Matrix}
\end{table}
After reading the stream we can calculate the estimators for joint and marginal probabilities,
$$\Tilde{s}_{i,j}=\frac{C_{i,j}}{N},\qquad \Tilde{p}_{i}=\frac{\sum C_{i,\cdot}}{N},\qquad \Tilde{q}_j=\frac{\sum C_{\cdot, j}}{N}$$
and use them to estimate the $\ell_2$ difference using a similar fashion as equation \ref{eq:l2}.

\begin{algorithm}[H]
	\SetKwInOut{Input}{Input}
	\SetKwInOut{Output}{Output}
	\caption{Approximating $\lVert s-r\rVert^2$ (Counter Matrix Algorithm)} \label{Alg:cm}
	\BlankLine
	\Input{$A$, the size of counter matrix}
	\Output{Estimation for $\lVert s-r\rVert^2$}
	\BlankLine
	$C\gets \mathbf{0}_{A\times A}$; \tcp{An $A\times A$ matrix of zeros}
	$h_1,h_2\gets$ random hash functions that maps $[n]\rightarrow [A]$;\\
	\For{each pair $(i,j)$ in stream}{
	    $C_{h_1(i),h_2(j)}\gets C_{h_1(i),h_2(j)}+1$; \tcp{increment the counter}
	}
	$\Upsilon\gets \frac{1}{(1-1/A)^2} \sum_{x,y\in[A]}\left(\frac{C_{x,y}}{N}-\frac{\sum C_{x,\cdot}\sum C_{\cdot,y}}{N^2}\right)^2$;\\
	\Return $\Upsilon$;
	\BlankLine
\end{algorithm}
Similar to the FM++ algorithm, we run the counter matrix algorithm $B$ times and take the median of all $\Upsilon$ as our final result.
\subsection{Analysis}

We can prove that our resulting algorithm 
\begin{itemize}
    \item is correct if $X$ and $Y$ are actually independent (theorem \ref{theorem:independent}),
    \item can give an unbiased estimate for $\ell_2$ difference (lemma \ref{theorem:l2}), and
    \item has bounded variance (lemma \ref{lemma:var}).
\end{itemize}
We want to prove the correctness of our method. First we show that \textbf{the estimator is correct when $X$ and $Y$ are actually independent}.
\begin{theorem}[Independent Case]\label{theorem:independent}
If $X$ and $Y$ are indeed independent, with sufficiently large $N$, we will have the estimated $\ell_1$ and $\ell_2$ difference be 0, i.e. $|\Tilde{\Delta}|=\lVert\Tilde{\Delta}\rVert =0$, where $\Tilde{\Delta}_{i,j}=\Tilde{s}_{i,j}-\Tilde{p}_{i}\Tilde{q}_j$.
\end{theorem}
\begin{proof}
Let us denote $h_1^{-1}(i):=\{x:h_1(x)=i\}$ the set of all $x$ that maps to $i$, and similarly $h_2^{-1}(j):=\{y:h_2(y)=j\}$. Let $h^{-1}(i,j):=h_1^{-1}(i)\times h_2^{-1}(j)$ be all the pairs $(x,y)$ that maps to $(i,j)$. Then for each $(i,j)\in[A]^2$, we will have the value of the counter be
$$C_{i,j}=\sum_{(x,y)\in h^{-1}(i,j)}n_{x,y}$$
and the sum of the corresponding row and column be,
$$\sum C_{i,\cdot}=\sum_{x\in h_1^{-1}(i)}n_{x,\cdot},\qquad \sum C_{\cdot,j}=\sum_{y\in h_2^{-1}(y)}n_{\cdot,y}$$
Therefore we will have $\Tilde{p}_i\Tilde{q}_j=\sum_{(x,y)\in h^{-1}(i,j)}n_{x,\cdot}n_{\cdot,y}/N^2$. Now, for sufficiently large $N$, we will have $n_{x,y}\rightarrow Ns_{x,y}, \sum n_{x,\cdot}\rightarrow Np_x$ and $\sum n_{\cdot,y}\rightarrow Nq_y$. Therefore,
$$\Tilde{\Delta}_{i,j}=\Tilde{s}_{i,j}-\Tilde{p}_{i}\Tilde{q}_j=\sum_{(x,y)\in h^{-1}(i,j)}\frac{n_{x,y}}{N}-\frac{n_{x,\cdot}n_{\cdot,y}}{N^2}\rightarrow \sum_{(x,y)\in h^{-1}(i,j)}\underbrace{s_{x,y}-p_xq_y}_{=0}=0$$
for all $i,j$. Therefore the $\ell_1$ and $\ell_2$ norm will be 0.
\end{proof}

If $X$ and $Y$ are correlated, however, things will get complicated. For $\ell_2$ difference case, let us consider the $\ell_2$ difference squared, i.e.
    $$\lVert\hat{\Delta}\rVert^2 =\sum_{i,j\in[n]}(\hat{\Delta}_{i,j})^2=\sum_{i,j\in[n]}(\hat{s}_{i,j}-\hat{p}_i\hat{q}_j)^2$$
    In the $A\times A$ counter matrix case, some of the $\hat{s}_{i,j}-\hat{p}_i\hat{q}_j$ will be added together then squared,
    \begin{eqnarray*}
    \lVert\Tilde{\Delta}\rVert^2 &=&\sum_{i,j\in[A]}(\Tilde{s}_{i,j}-\Tilde{p}_i\Tilde{q}_j)^2=\sum_{i,j\in[A]}\left(\sum_{(x,y)\in h^{-1}(i,j)}\hat{s}_{x,y}-\hat{p}_x\hat{q}_j\right)^2=\sum_{i,j\in[A]}\left(\sum_{(x,y)\in h^{-1}(i,j)}\hat{\Delta}_{x,y}\right)^2\\
    &=&\underbrace{\sum_{i,j\in[n]}(\hat{\Delta}_{i,j})^2}_{=\lVert\hat{\Delta}\rVert^2}+\underbrace{\sum_{\substack{h(i,j)=h(k,l)\\(i,j)\ne(k,l)}}\hat{\Delta}_{i,j}\hat{\Delta}_{k,l}}_{=:X}=\lVert\hat{\Delta}\rVert^2+X
    \end{eqnarray*}
    And we can see that the $\ell_2$ estimate given by our $A\times A$ matrix can be expressed by $\lVert\hat{\Delta}\rVert^2$ plus a term $X$, where $X$ is just the sum of product of $\hat{\Delta}$ of all the colliding pairs in our hash function. We thus want to study the distribution of $X$.
    
    Before going into this, we want to show a property of $\hat{\Delta}$.
    \begin{lemma}[Property of $\hat{\Delta}$]\label{lemma:sum}
    The sum of each row or each column of $\hat{\Delta}$ must be 0.
    \end{lemma}
    \begin{proof}
    For row $i$, for example, the sum of all elements on row $i$ would be
    \begin{eqnarray*}
    \sum \hat{\Delta}_{i,\cdot}&=&\sum_{y\in[A]}\sum_{x\in [A]}(\hat{s}_{x,y}-\hat{p}_x\hat{q}_y)=\sum_{x\in [A]}\underbrace{\sum_{y\in[A]}\hat{s}_{x,y}}_{=\hat{p}_x}-\sum_{x\in [A]}\hat{p}_x\underbrace{\sum_{y\in[A]}\hat{q}_y}_{=1}\\
    &=&\sum_{x\in [A]}\hat{p}_x-\sum_{x\in [A]}\hat{p}_x=0
    \end{eqnarray*}
    The column case can be proved by a similar fashion.
    \end{proof}
    The following result immediately follows from the above lemma.
    \begin{corollary}\label{cor:sum}
    All elements in $\hat{\Delta}$ sums to 0.
    \end{corollary}
    Now, let us investigate the expectation of $X$. The probability of the term $\hat{\Delta}_{i,j}\hat{\Delta}_{k,l}$ appear in $X$ is equal to the probability of $(k,l)$ having the same hash values as $(i,j)$ (collide).
    $$\Pr[h(k,l)=h(i,j)]=\left\{\begin{array}{ll}
         1/A^2&\text{if }i\ne k, j\ne l  \\
         1/A&\text{ if}i=k\text{ or }j=l 
    \end{array}\right.$$
    since if $i=k$ or $j=l$, we only need to ensure the other dimension to have the same hash values, while if $i\ne k, j\ne l$, we need both pairs to have the same hash values.
    \begin{table}[H]
    \centering
\begin{tabular}{rccccccc}
                              & 1                      & 2                      & $\dots$                      & $j$                                                                          & $\dots$                    & $n-1$                    & $n$                    \\ \cline{2-8} 
\multicolumn{1}{r|}{1}        & \multicolumn{3}{c|}{\cellcolor[HTML]{EFEFEF}}                                  & \multicolumn{1}{c|}{\cellcolor[HTML]{C0C0C0}}                                & \multicolumn{3}{c|}{\cellcolor[HTML]{EFEFEF}}                                  \\
\multicolumn{1}{r|}{2}        & \multicolumn{3}{c|}{\cellcolor[HTML]{EFEFEF}}                                  & \multicolumn{1}{c|}{\cellcolor[HTML]{C0C0C0}}                                & \multicolumn{3}{c|}{\cellcolor[HTML]{EFEFEF}}                                  \\
\multicolumn{1}{r|}{$\vdots$} & \multicolumn{3}{c|}{\multirow{-3}{*}{\cellcolor[HTML]{EFEFEF}$\frac{1}{A^2}$}} & \multicolumn{1}{c|}{\multirow{-3}{*}{\cellcolor[HTML]{C0C0C0}$\frac{1}{A}$}} & \multicolumn{3}{c|}{\multirow{-3}{*}{\cellcolor[HTML]{EFEFEF}$\frac{1}{A^2}$}} \\ \cline{2-8} 
\multicolumn{1}{r|}{$i$}      & \multicolumn{3}{c|}{\cellcolor[HTML]{C0C0C0}$\frac{1}{A}$}                     & \multicolumn{1}{c|}{\cellcolor[HTML]{656565}}                                & \multicolumn{3}{c|}{\cellcolor[HTML]{C0C0C0}$\frac{1}{A}$}                     \\ \cline{2-8} 
\multicolumn{1}{r|}{$\vdots$} & \multicolumn{3}{c|}{\cellcolor[HTML]{EFEFEF}}                                  & \multicolumn{1}{c|}{\cellcolor[HTML]{C0C0C0}}                                & \multicolumn{3}{c|}{\cellcolor[HTML]{EFEFEF}}                                  \\
\multicolumn{1}{r|}{$n-1$}    & \multicolumn{3}{c|}{\cellcolor[HTML]{EFEFEF}}                                  & \multicolumn{1}{c|}{\cellcolor[HTML]{C0C0C0}}                                & \multicolumn{3}{c|}{\cellcolor[HTML]{EFEFEF}}                                  \\
\multicolumn{1}{r|}{$n$}      & \multicolumn{3}{c|}{\multirow{-3}{*}{\cellcolor[HTML]{EFEFEF}$\frac{1}{A^2}$}} & \multicolumn{1}{c|}{\multirow{-3}{*}{\cellcolor[HTML]{C0C0C0}$\frac{1}{A}$}} & \multicolumn{3}{c|}{\multirow{-3}{*}{\cellcolor[HTML]{EFEFEF}$\frac{1}{A^2}$}} \\ \cline{2-8} 
\end{tabular}
\caption{Illustration of probability that a pair collides with (has the same hash values as) $(i,j)$.}
\end{table}
Using lemma \ref{lemma:sum} and corollary \ref{cor:sum}, the expectation of the sum of all terms containing $\hat{\Delta}_{i,j}$ in $X$ is then
\begin{eqnarray*}
&&\hat{\Delta}_{i,j}\left[\frac{1}{A^2}\underbrace{\sum_{(i,j)\ne (k,l)}\hat{\Delta}_{k,l}}_{=-\hat{\Delta}_{i,j}}+\left(\frac{1}{A}-\frac{1}{A^2}\right)\underbrace{\sum_{i\ne k}\hat{\Delta}_{k,j}}_{=-\hat{\Delta}_{i,j}}+\left(\frac{1}{A}-\frac{1}{A^2}\right)\underbrace{\sum_{j\ne l}\hat{\Delta}_{i,l}}_{=-\hat{\Delta}_{i,j}}\right]\\
&=&-\left(\frac{2}{A}-\frac{1}{A^2}\right)\hat{\Delta}_{i,j}^2
\end{eqnarray*}
Therefore, the expectation of $X$ is then
$$\E[X]=-\left(\frac{2}{A}-\frac{1}{A^2}\right)\sum_{i,j\in[n]}\hat{\Delta}_{i,j}^2=-\left(\frac{2}{A}-\frac{1}{A^2}\right)\lVert \hat{\Delta}\rVert^2$$
Meaning that
$$\E[\lVert\Tilde{\Delta}\rVert^2]=\left(1-\frac{2}{A}+\frac{1}{A^2}\right)\lVert \hat{\Delta}\rVert^2=\left(1-\frac{1}{A}\right)^2\lVert \hat{\Delta}\rVert^2$$
This gives rise to our theorem,
\begin{lemma}[Unbiased Estimator]\label{theorem:l2}
$\lVert \Tilde{\Delta}\rVert^2/(1-1/A)^2$ is going to be an unbiased estimator for $\lVert \hat{\Delta}\rVert^2$.
\end{lemma}
This is a neat result, as it is not dependent on $n$ or $N$. This lemma indicates that we can, again, create multiple $A\times A$ counter matrices, divide the resulting squared $\ell_2$ difference with $(1-1/A)^2$, and take the mean to get an unbiased estimation.

\textbf{Remark.} We should have $A>1$ for this algorithm to be meaningful. If we set $A=1$, the expectation $\E[\lVert\Tilde{\Delta}\rVert^2]=0$. In fact, if we only use 1 number, then $s_{1,1}=p_1=q_1=1$, so the resulting norm will always be zero.

We can also prove that the variance is also bounded,
\begin{lemma}[Bounded Variance]\label{lemma:var}
The variance of estimator $\Var[\lVert \Tilde{\Delta}\rVert^2]\leq 8\lVert\hat{\Delta}\rVert^4/A$.
\end{lemma}
The proof is pretty long, so we put it in the appendix. This result shows that we can utilize Chebyshev's inequality to limit the error rate.
\begin{theorem}[Counter Matrix Algorithm]
There exists a one-pass, $\mathcal{O}(\epsilon^{-4}\log\delta^{-1})$-space algorithm to calculate an estimator $\Upsilon$ such that $\Upsilon\in (1\pm\epsilon)\lVert \hat{\Delta}\rVert^2$ with probability at least $1-\delta$.
\end{theorem}
\begin{proof}
We can take $A=32/\epsilon^2$ such that by Chebyshev's Inequality, we have can have $\Upsilon\in (1\pm\epsilon)\lVert \hat{\Delta}\rVert^2$ with probability at least $1/4$, and we can take $B=32\log (2/\delta)$ to ensure that the median is a correct answer with probability at least $1-\delta$. In total we need $A^2B=\mathcal{O}(\epsilon^{-4}\log\delta^{-1})$ space.
\end{proof}
This is, however, a pretty loose upper bound. We can see our experiment that the variance $\Var[\Upsilon]=\mathcal{O}(\lVert\hat{\Delta}\rVert^4/A^2)$ when $A$ is large. Therefore, in reality we can achieve $\mathcal{O}(\epsilon^{-2}\log\delta^{-1})$-space, which is similar to the sketching of sketches algorithm.
\subsection{Implementation}
For now, we have completed the implementation of this algorithm estimating the $\ell_2$ difference. We choose the hash functions to be
$$h_1(x)=(ax+b)\mod p \mod A,\qquad h_2(x)=(cx+d)\mod p \mod A$$
where $p$ is a random prime that is significantly larger than $A$ and $a,b,c,d<p$ are randomly chosen. We only need to compute two hash functions and store the 4 parameters for each counter matrix. Comparing this with at least 8 parameters for each run of the sketching of sketches algorithm, our algorithm is quite fast and at least 2 times more memory-saving.

\section{Results}\label{sec:experiment}
The code and benchmark for our paper is given here: \url{github.com/GZHoffie/CS5234-mini-project}.
\subsection{Data Set Generation}\label{dataset}

We need to generate two types of distribution for sample stream:

\begin{enumerate}
    \item \textbf{$X$ and $Y$ are independent}: We first generate 2 marginal probability arrays $p_x$ and $q_y$ of size $n$ with each array sums to 1. Then we generated each element pairs using the probability arrays obtained.
    \item \textbf{$X$ and $Y$ are correlated}: Generate the probability matrix $s_{i,j}$ for each $i,j\in[n]$ and normalize it such that all values sums up to 1. Then we generated each element pairs using the probability matrix obtained. 
\end{enumerate}

We use two types of distributions:
\begin{enumerate}
    \item \textbf{Random distribution}: where $s_{i,j}$ or $p_i$ and $q_j$ are randomly generated.
    \item \textbf{Zipfian distribution}: where $s_{i,j}$ or $p_i$ and $q_j$ follows Zipf's Law \cite{zipf2016human}, where the second largest probability is $1/2$ of the largest probability, and the third largest is $1/3$, so on and so forth.
\end{enumerate}

We also use a $n\times n$ counter matrix to calculate the values for the exact estimate of the $\ell_2$ difference, $\lVert\hat{\Delta}\rVert$, and store them to calculate the multiplicative error.

In the following experiments, we used 4 generated data sets with $N=1000000$ and $n=10000$, produced by our data generation scripts  mentioned in \ref{dataset}. Each data set is generated with a different combination distribution type and whether the two stream are independent or not.

\begin{enumerate}
    \item Random distribution, dependent. $\lVert\hat{\Delta}\rVert\approx 0.0010$.
    \item Random distribution, independent. $\lVert\hat{\Delta}\rVert\approx 0.00099$ (similar to the dependent case since real $\lVert\Delta\rVert\approx 5.77\times 10^{-5}$ is very small).
    \item Zipfian distribution, dependent. $\lVert\hat{\Delta}\rVert\approx 0.065$.
    \item Zipfian distribution, independent. $\lVert\hat{\Delta}\rVert\approx 0.00097$.
\end{enumerate}

\subsection{Metrics}

We can compare the estimated result with the exact solution. In particular, we used the multiplicative error,
$$\epsilon=\left|1-\frac{\Upsilon-\lVert r-s\rVert}{\lVert r-s\rVert}\right|$$
to check how far our estimator $\Upsilon$ is from the actual $\ell_2$ difference.

\subsection{Experiments}

We designed 2 experiments to evaluate our counter matrix algorithm.

\begin{enumerate}
    \item Given fixed space, determine the trade-off of parameters $A$ and $B$ used in the algorithm
    \item Given fixed space, evaluate the multiplicative error by counter matrix algorithm and algorithm \ref{Alg:l2}
\end{enumerate}

\subsubsection{Effects of Parameters $A$ and $B$}

In this experiment, for each data set we ran a set of sub-experiments. By doing a grid search of changing $A$ and $B$ values, we can find out how the multiplicative error change. We have selected 5 $A$ values and 5 $B$ values on purpose such that if we arrange those values in a table where the cells in the same diagonal direction from top-right to bottom-left (indicated using the same color in the table) will use the same amount of space, $A^2B$.

Since storing the exact probability table for the stream takes $\mathcal{O}(n^2)$ space, the counter matrix algorithm will only make sense if $A^2B= o(n^2)$. Also, for each $A$ and $B$, type of distribution and dependence, we repeat the experiment 5 times. The results below shows the average multiplicative error for the 5 runs. 

\begin{table}[H]
\centering
\begin{tabular}{l c c c c c}
\hline\hline
&A=2 & A=4 & A=8 & A=16 & A=32 \\ [0.5ex] % inserts table %heading
\hline
B=1  & \colorbox{red!30}{0.450537}  & \colorbox{green!30}{0.192048}  & \colorbox{blue!30}{0.053852} & \colorbox{teal!30}{0.023916} & \colorbox{lime!30}{0.015661} \\
B=4  & \colorbox{green!30}{0.414036} & \colorbox{blue!30}{0.141648}  & \colorbox{teal!30}{0.053533} & \colorbox{lime!30}{0.025423} & \colorbox{brown!30}{0.007171} \\
B=16 & \colorbox{blue!30}{0.239866} & \colorbox{teal!30}{0.063814} & \colorbox{lime!30}{0.026485} & \colorbox{brown!30}{0.009675} & \colorbox{cyan!30}{0.004532} \\
B=64 & \colorbox{teal!30}{0.297715} & \colorbox{lime!30}{0.052855} & \colorbox{brown!30}{0.021604} & \colorbox{cyan!30}{0.006910} & \colorbox{orange!30}{0.005641} \\
B=256& \colorbox{lime!30}{0.288948} & \colorbox{brown!30}{0.057999} & \colorbox{cyan!30}{0.018541} & \colorbox{orange!30}{0.008588} & \colorbox{magenta!30}{0.001205} \\
\hline
\end{tabular}
\label{table:random_dependent}
\caption{Multiplicative error with Random Distribution, Dependent Stream}
\end{table}

\begin{table}[H]
\centering
\begin{tabular}{l c c c c c}
\hline\hline
&A=2 & A=4 & A=8 & A=16 & A=32 \\ [0.5ex] % inserts table %heading
\hline
B=1  & \colorbox{red!30}{0.401140}  & \colorbox{green!30}{0.292316}  & \colorbox{blue!30}{0.100423} & \colorbox{teal!30}{0.051171} & \colorbox{lime!30}{0.017768} \\
B=4  & \colorbox{green!30}{0.224761} & \colorbox{blue!30}{0.118114}  & \colorbox{teal!30}{0.019799} & \colorbox{lime!30}{0.012302} & \colorbox{brown!30}{0.011697} \\
B=16 & \colorbox{blue!30}{0.286748} & \colorbox{teal!30}{0.052774} & \colorbox{lime!30}{0.008510} & \colorbox{brown!30}{0.012821} & \colorbox{cyan!30}{0.008939} \\
B=64 & \colorbox{teal!30}{0.328810} & \colorbox{lime!30}{0.025917} & \colorbox{brown!30}{0.022485} & \colorbox{cyan!30}{0.007506} & \colorbox{orange!30}{0.002322} \\
B=256& \colorbox{lime!30}{0.295032} & \colorbox{brown!30}{0.011417} & \colorbox{cyan!30}{0.021785} & \colorbox{orange!30}{0.003699} & \colorbox{magenta!30}{0.002610} \\
\hline
\end{tabular}
\label{table:random_independent}
\caption{Multiplicative error with Random Distribution, Independent Stream}
\end{table}

\begin{table}[H]
\label{table:zipfian_dependent}
\centering
\begin{tabular}{l c c c c c}
\hline\hline
&A=2 & A=4 & A=8 & A=16 & A=32 \\ [0.5ex] % inserts table %heading
\hline
B=1  & \colorbox{red!30}{0.561978}  & \colorbox{green!30}{0.116835}  & \colorbox{blue!30}{0.053566} & \colorbox{teal!30}{0.026087} & \colorbox{lime!30}{0.005170} \\
B=4  & \colorbox{green!30}{0.241000} & \colorbox{blue!30}{0.055140}  & \colorbox{teal!30}{0.040690} & \colorbox{lime!30}{0.014915} & \colorbox{brown!30}{0.005655} \\
B=16 & \colorbox{blue!30}{0.203792} & \colorbox{teal!30}{0.037948} & \colorbox{lime!30}{0.012536} & \colorbox{brown!30}{0.007444} & \colorbox{cyan!30}{0.003694} \\
B=64 & \colorbox{teal!30}{0.230859} & \colorbox{lime!30}{0.037576} & \colorbox{brown!30}{0.005997} & \colorbox{cyan!30}{0.002579} & \colorbox{orange!30}{0.001762} \\
B=256& \colorbox{lime!30}{0.203329} & \colorbox{brown!30}{0.028557} & \colorbox{cyan!30}{0.009255} & \colorbox{orange!30}{0.001563} & \colorbox{magenta!30}{0.002044} \\
\hline
\end{tabular}
\caption{Multiplicative error with Zipfian Distribution, Dependent Stream}
\end{table}

\begin{table}[H]
\centering
\begin{tabular}{l c c c c c}
\hline\hline
&A=2 & A=4 & A=8 & A=16 & A=32 \\ [0.5ex] % inserts table %heading
\hline
B=1  & \colorbox{red!30}{0.567727}  & \colorbox{green!30}{0.117505}  & \colorbox{blue!30}{0.082284} & \colorbox{teal!30}{0.038022} & \colorbox{lime!30}{0.014242} \\
B=4  & \colorbox{green!30}{0.465189} & \colorbox{blue!30}{0.039366}  & \colorbox{teal!30}{0.031472} & \colorbox{lime!30}{0.023964} & \colorbox{brown!30}{0.014675} \\
B=16 & \colorbox{blue!30}{0.352457} & \colorbox{teal!30}{0.068149} & \colorbox{lime!30}{0.017812} & \colorbox{brown!30}{0.016407} & \colorbox{cyan!30}{0.011978} \\
B=64 & \colorbox{teal!30}{0.311025} & \colorbox{lime!30}{0.040865} & \colorbox{brown!30}{0.025680} & \colorbox{cyan!30}{0.018174} & \colorbox{orange!30}{0.011105} \\
B=256& \colorbox{lime!30}{0.285863} & \colorbox{brown!30}{0.029304} & \colorbox{cyan!30}{0.013811} & \colorbox{orange!30}{0.013607} & \colorbox{magenta!30}{0.010472}\\
\hline
\end{tabular}
\label{table:zipfian_independent}
\caption{Multiplicative error with Zipfian Distribution, Independent Stream}
\end{table}

We can summarize the following findings from this experiment from data in 4 tables:

\begin{enumerate}
    \item The counter matrix algorithm is effective, and achieve a small multiplicative error when $A\geq 32$.
    \item Regardless of distribution type, and whether the streams are independent or not, the multiplicative error will be lower as $A$ or $B$ increases.
    \item Given a fixed space where $A^2B$ are same (indicated by cells with same color in the tables above), increasing $A$ generally results in lower multiplicative error.
\end{enumerate}

\subsubsection{Comparison with Sketching of Sketches Algorithm}
In this experiment, we try to compare counter matrix algorithm with the sketching of sketches algorithm given the same amount of space. More specifically, we choose multiple values for $A$, and run algorithm \ref{Alg:cm} with parameter $A$, and $A^2$ copies of algorithm \ref{Alg:l2} so that both algorithms use roughly $A^2$ space, and compare the multiplicative error. For each $A$, we run 10 experiments and take the average error.

\begin{figure}[H] 
  \begin{subfigure}[b]{0.5\linewidth}
    \centering
    \includegraphics[width=0.9\linewidth]{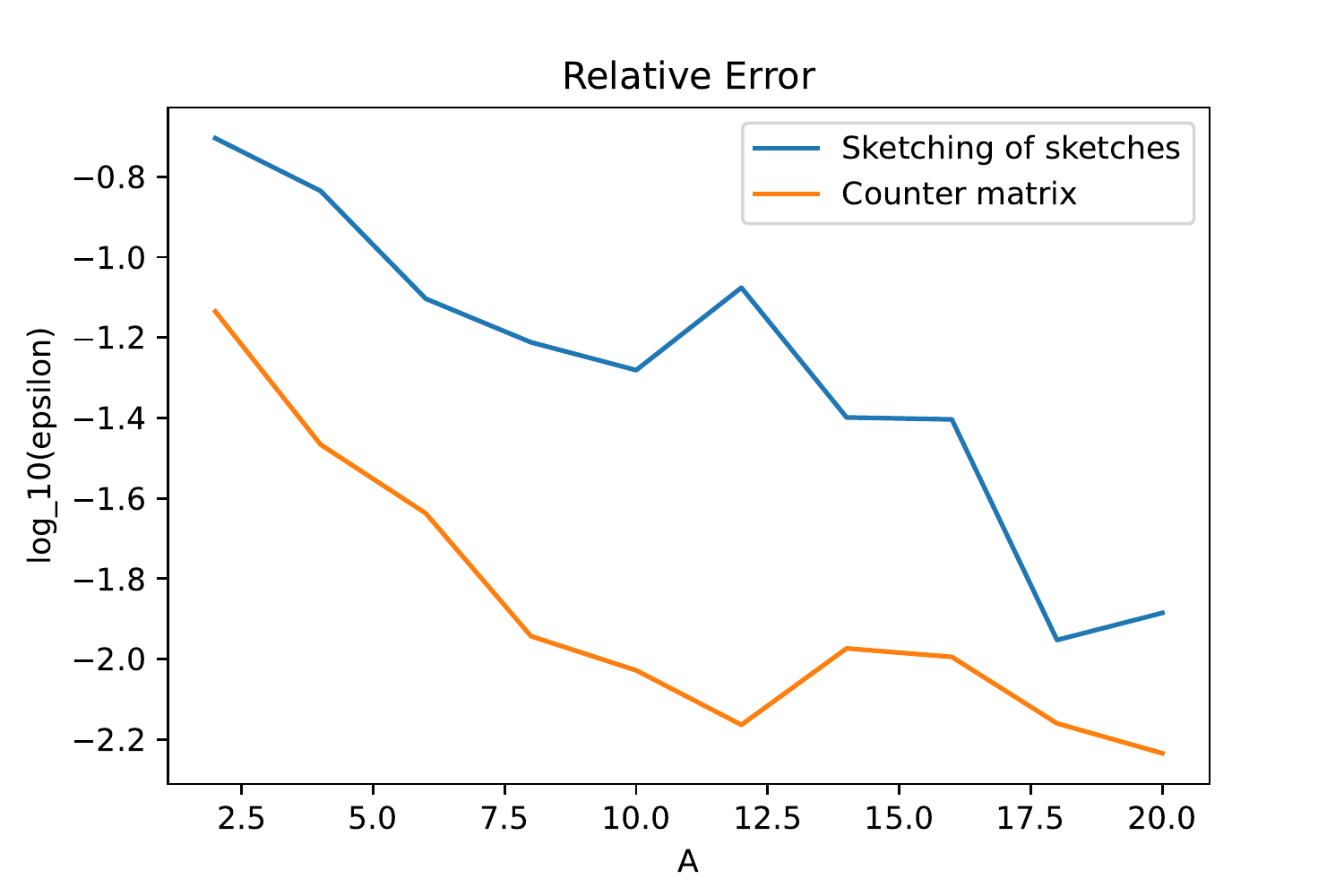} 
    \caption{Random Distribution, Dependent Stream} 
    \label{fig7:a} 
    \vspace{4ex}
  \end{subfigure}%% 
  \begin{subfigure}[b]{0.5\linewidth}
    \centering
    \includegraphics[width=0.9\linewidth]{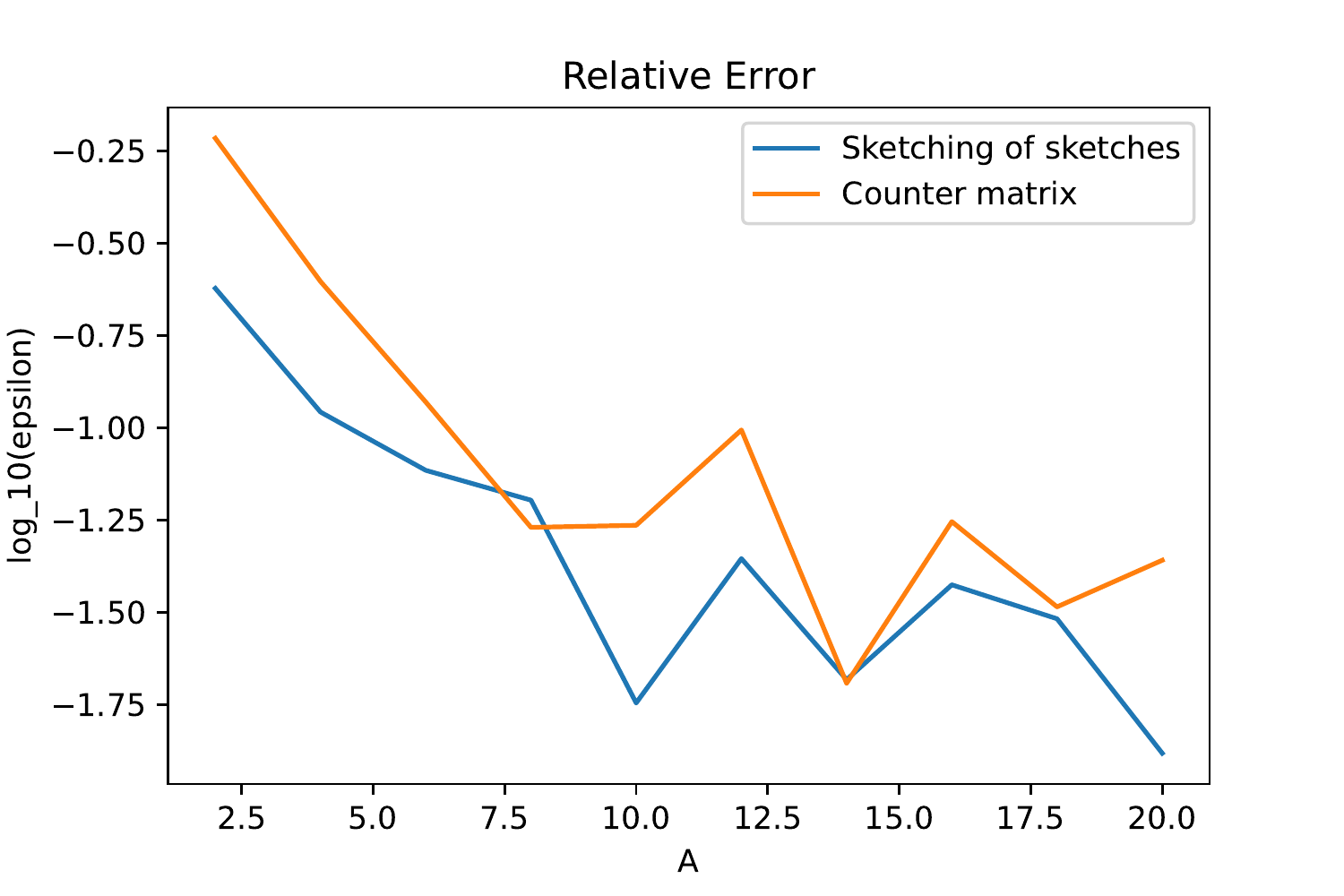} 
    \caption{Random Distribution, Inependent Stream} 
    \label{fig7:b} 
    \vspace{4ex}
  \end{subfigure} 
  \begin{subfigure}[b]{0.5\linewidth}
    \centering
    \includegraphics[width=0.9\linewidth]{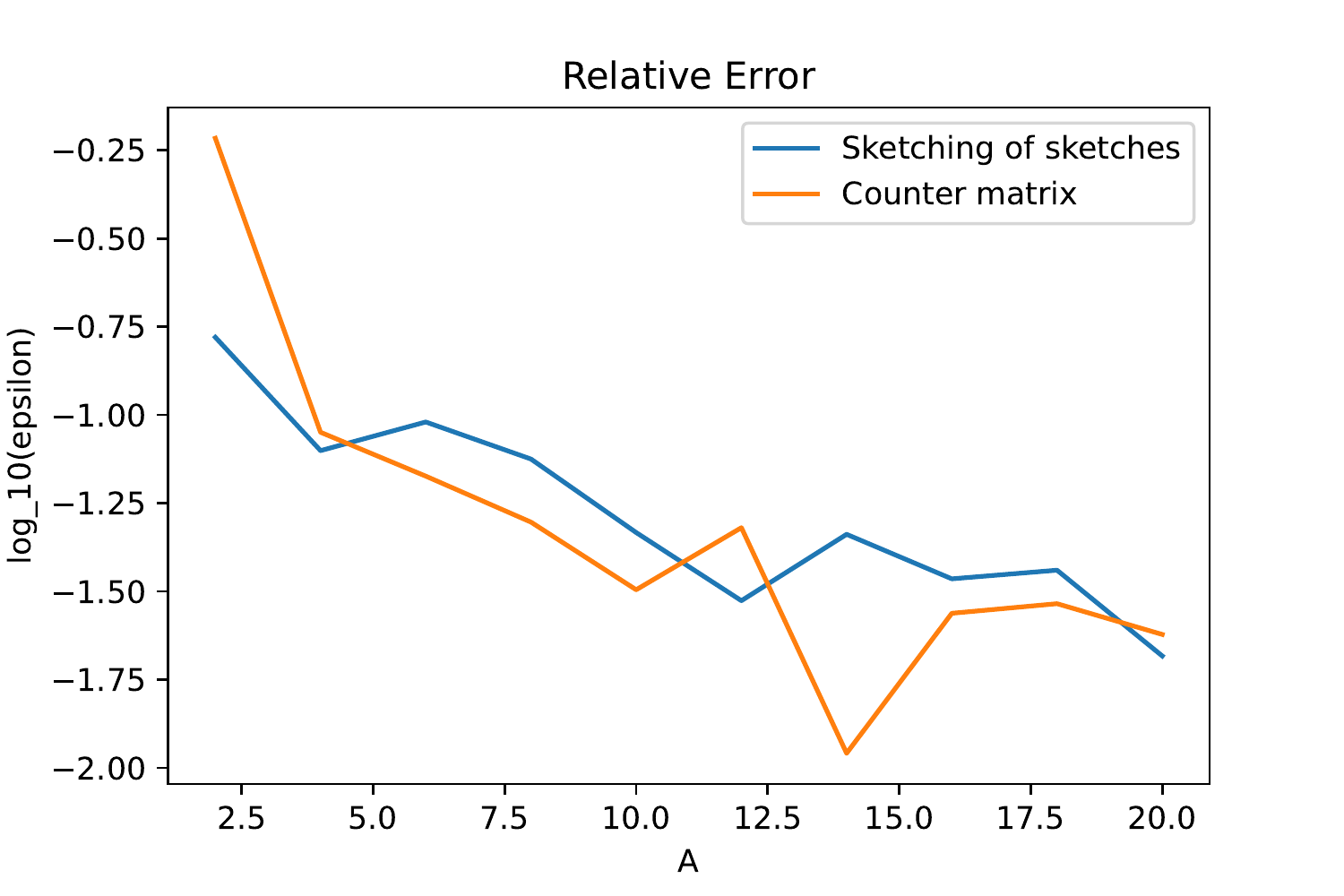} 
    \caption{Zipfian Distribution, Dependent Stream} 
    \label{fig7:c} 
  \end{subfigure}%%
  \begin{subfigure}[b]{0.5\linewidth}
    \centering
    \includegraphics[width=0.9\linewidth]{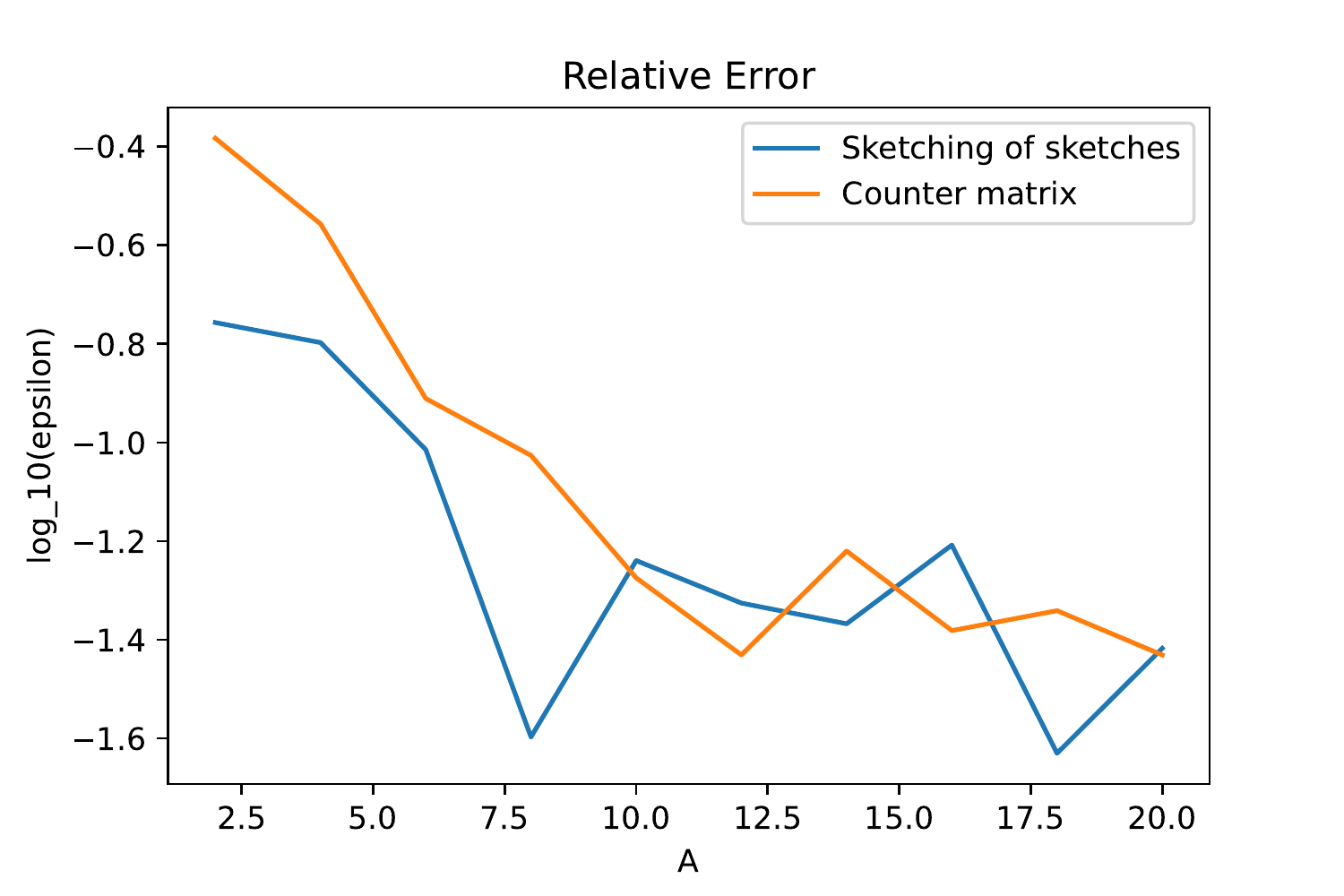} 
    \caption{Zipfian Distribution, Inependent Stream} 
    \label{fig7:d} 
  \end{subfigure} 
  \caption{Comparison of log multiplicative error of two algorithms}
  \label{fig7} 
\end{figure}
We can see that
\begin{enumerate}
    \item In both cases, increasing $A$ will generally lead to smaller multiplicative error, which is as expected.
    \item Even for small $A$ ($2\leq A\leq 20$), the multiplicative error of the counter matrix algorithm is similar or smaller than the sketching of sketches algorrithm.
\end{enumerate}
Considering the fact that the sketching of sketches algorithm actually uses an additional $(8A)^2$ space to store the hash parameters (or vectors taking $(nA)^2$ space in the original implementation in the paper), and our counter matrix algorithm only need to store $4$ hash parameters for each run, we can see that our algorithm is actually uses 16x less space, and can get a relatively good result at the same time.
\section{Conclusion}

In this mini-project, we analyzed and implemented 2 algorithms proposed in the paper by Indyk and McGregor \cite{inproceedings}:

\begin{enumerate}
    \item \textbf{For $\ell_2$ difference}: A 1-pass, $\mathcal{O}(\epsilon^{-2}\log\delta^{-1})$-space estimation that achieve $(1+\epsilon,\delta)$-approximation.
    \item \textbf{For $\ell_1$ difference}: A 1-pass, $\mathcal{O}(\log\delta^{-1})$-space estimation that achieve  $(\mathcal{O}(\log n),\delta)$-approximation.
\end{enumerate}

The implementation of the algorithm to estimate $\ell_2$ difference we achieved satisfactory results, however the algorithm to estimate $\ell_2$ difference produced error way beyond the estimated error stated in the paper. As an extension to estimate $\ell_2$ difference, we proposed a different counter matrix algorithm that aimed to achieve lower multiplicative error that the one proposed in the paper by Indyk and McGregor \cite{inproceedings}. From our analysis and experiments above, we have showed that the counter matrix algorithm is able to achieve better approximation and produce lower or similar multiplicative error than the algorithm proposed in the paper given same (or even less) amount of space.

\newpage
\section*{Appendix}
Here we try to calculate $\Var[\lVert\Tilde{\Delta}\rVert^2]$. Since we have $\lVert\Tilde{\Delta}\rVert^2=\lVert\hat{\Delta}\rVert^2+X$, and we consider $\lVert\hat{\Delta}\rVert^2$ as a constant, then $\Var[\lVert\Tilde{\Delta}\rVert^2]=\Var[X]=\E[X^2]-\E[X]^2$. We have
$$X=\sum_{\substack{h(i,j)=h(k,l)\\(i,j)\ne(k,l)}}\hat{\Delta}_{i,j}\hat{\Delta}_{k,l}=\sum_{i,j,k,l\in[n]}\delta_{i,j,k,l}\hat{\Delta}_{i,j}\hat{\Delta}_{k,l}$$
where $\delta_{i,j,k,l}=1$ if $h(i,j)=h(k,l)$ and $(i,j)\ne(k,l)$, and 0 otherwise. We have calculated that
$$\E[\delta_{i,j,k,l}]=\left\{\begin{array}{ll}
         1/A^2&\text{if }i\ne k, j\ne l  \\
         1/A&\text{if }i=k\text{ or }j=l 
    \end{array}\right.$$
Now let us consider imagine that there is another hash function $h'$ that is independent of $h$, and a random variable $X'$, where
$$X'=\sum_{i',j',k',l'\in[n]}\delta'_{i',j',k',l'}\hat{\Delta}_{i',j'}\hat{\Delta}_{k',l'}$$
where $\delta'_{i,j,k,l}=1$ if $h'(i,j)=h'(k,l)$ and $(i,j)\ne(k,l)$, and 0 otherwise. Then $X$ is independent of $X'$ and
$$\E[X]^2=\E[XX']=\E\left[\sum_{i,j,k,l\in[n]}\sum_{i',j',k',l'\in[n]}\delta_{i,j,k,l}\delta'_{i',j',k',l'}\hat{\Delta}_{i,j}\hat{\Delta}_{k,l}\hat{\Delta}_{i',j'}\hat{\Delta}_{k',l'}\right]$$
On the other hand,
$$\E[X^2]=\E\left[\sum_{i,j,k,l\in[n]}\sum_{i',j',k',l'\in[n]}\delta_{i,j,k,l}\delta_{i',j',k',l'}\hat{\Delta}_{i,j}\hat{\Delta}_{k,l}\hat{\Delta}_{i',j'}\hat{\Delta}_{k',l'}\right]$$
So we just need to find those $(i,j,k,l,i',j',k',l')$ where $\delta_{i,j,k,l}\delta_{i',j',k',l'}\ne \delta_{i,j,k,l}\delta'_{i',j',k',l'}$ to calculate the variance. There are several cases we need to discuss.
\begin{enumerate}
    \item \textbf{If $(i,j)$ and $(k,l)$ are on the same row, i.e. $i=k$, and we have $\delta_{i,j,k,l}=1$}. It implies that $h_2(j)=h_2(l)$, and therefore if $j'=j$ and $l'=l$ or $j'=l$ and $l'=j$, the probability of $h(i',j')=h(k',l')$ would be increased. 
    
    For example, if we know $h(i,j)=h(i,l)$, then the probability of $h(1,j)=h(1,l)$ will be 1 (which originally was $1/A$), and the probability of $h(1,j)=h(2,l)$ will be $1/A$ (which originally was $1/A^2$).
    
    We can then calculate the expected total difference in the variance for each pair of $(i,j)$ and $(i,l)$,
    \begin{eqnarray*}
    &&\frac{1}{A}\hat{\Delta}_{i,j}\hat{\Delta}_{i,l}\left[\left(1-\frac{1}{A}\right)\sum_{i'=1}^n (\hat{\Delta}_{i',j}\hat{\Delta}_{i',l}+\hat{\Delta}_{i',l}\hat{\Delta}_{i',j})+\left(\frac{1}{A}-\frac{1}{A^2}\right)\sum_{i'\ne k'}(\hat{\Delta}_{i',j}\hat{\Delta}_{k',l}+\hat{\Delta}_{i',l}\hat{\Delta}_{k',j})\right]\\
    &=&\frac{2}{A}\hat{\Delta}_{i,j}\hat{\Delta}_{i,l}\left[\left(1-\frac{1}{A}\right)\sum_{i'=1}^n \hat{\Delta}_{i',j}\hat{\Delta}_{i',l}+\left(\frac{1}{A}-\frac{1}{A^2}\right)\sum_{i'\ne k'}\hat{\Delta}_{i',j}\hat{\Delta}_{k',l}\right]\\
    &=&\frac{2}{A}\hat{\Delta}_{i,j}\hat{\Delta}_{i,l}\left[\left(1-\frac{1}{A}\right)^2\sum_{i'=1}^n \hat{\Delta}_{i',j}\hat{\Delta}_{i',l}+\left(\frac{1}{A}-\frac{1}{A^2}\right)\underbrace{\sum_{i',k'\in[n]}\hat{\Delta}_{i',j}\hat{\Delta}_{k',l}}_{=0}\right]\\
    &=&\frac{2}{A}\hat{\Delta}_{i,j}\hat{\Delta}_{i,l}\left[\left(1-\frac{1}{A}\right)^2\sum_{i'=1}^n \hat{\Delta}_{i',j}\hat{\Delta}_{i',l}\right]
    \end{eqnarray*}
    Summing these all up across all pairs of $(i,j),(i,l)$, we will get
    $$\frac{2}{A}\left(1-\frac{1}{A}\right)^2\underbrace{\sum_{\substack{(i,j),(k,l),(i',j'),(k',l')\\ \text{form a rectangle}}}\hat{\Delta}_{i,j}\hat{\Delta}_{k,l}\hat{\Delta}_{i',j'}\hat{\Delta}_{k',l'}}_{=:\mathcal{R}}$$
    where in $\mathcal{R}$, each rectangle of same position and shape is counted 4 times.
    \item \textbf{If $(i,j)$ and $(k,l)$ are on the same column, i.e. $j=l$, and we have $\delta_{i,j,k,l}=1$}. This case is totally the same as the previous case, which add another $2\mathcal{R}(1-1/A)^2/A$ to our variance.
    \item \textbf{If $(i,j)$ and $(k,l)$ are on different column and different row, i.e. $i\ne k$ and $j\ne l$, and  $\delta_{i,j,k,l}=1$}, then
    \begin{itemize}
        \item If $\{(i',j'),(k',l')\}=\{(i,j),(k,l)\}$, then the probability of $h(i',j')=h(k',l')$ increased from $1/A^2$ to 1. This means that for each pair of $(i,j)$ and $(k,l)$, the following is added to the variance,
        $$\frac{1}{A^2}\hat{\Delta}_{i,j}\hat{\Delta}_{k,l}\left[2\left(1-\frac{1}{A^2}\right)\hat{\Delta}_{i,j}\hat{\Delta}_{k,l}\right]=\frac{2}{A^2}\left(1-\frac{1}{A^2}\right)\hat{\Delta}^2_{i,j}\hat{\Delta}^2_{k,l}$$
        Summing them up across all pairs gives us
        $$\frac{2}{A^2}\left(1-\frac{1}{A^2}\right)\sum_{\substack{i\ne k\\j\ne l}}\hat{\Delta}^2_{i,j}\hat{\Delta}^2_{k,l}\leq \frac{2}{A^2}\left(1-\frac{1}{A^2}\right)\lVert\hat{\Delta}\rVert^2\sum_{i,j\in [n]}\hat{\Delta}^2_{i,j}\leq \frac{2}{A^2}\left(1-\frac{1}{A^2}\right)\lVert\hat{\Delta}\rVert^4$$
        \item If $\{(i',j'),(k',l')\}=\{(i,l),(k,j)\}$, then the probability of $h(i',j')=h(k',l')$ also increased from $1/A^2$ to 1. Therefore we also add the following,
        $$\frac{1}{A^2}\hat{\Delta}_{i,j}\hat{\Delta}_{k,l}\left[2\left(1-\frac{1}{A^2}\right)\hat{\Delta}_{i,l}\hat{\Delta}_{k,j}\right]$$
        summing across all pairs gives us
        $$\frac{2}{A^2}\left(1-\frac{1}{A^2}\right)\mathcal{R}$$
        \item If $\{j',l'\}=\{j,l\}$, then if $i'=k'$ the probability of matching is increased from $1/A$ to 1, otherwise the probability increased from $1/A^2$ to $1/A$. This is similar to our first case, adding the following term to our variance,
        \begin{eqnarray*}
    &&\frac{1}{A^2}\hat{\Delta}_{i,j}\hat{\Delta}_{k,l}\left[2\left(1-\frac{1}{A}\right)\sum_{i'=1}^n \hat{\Delta}_{i',j}\hat{\Delta}_{i',l}+2\left(\frac{1}{A}-\frac{1}{A^2}\right)\sum_{i'\ne k'}\hat{\Delta}_{i',j}\hat{\Delta}_{k',l}\right]\\
    &=&\frac{2}{A^2}\hat{\Delta}_{i,j}\hat{\Delta}_{k,l}\left[\left(1-\frac{1}{A}\right)^2\sum_{i'=1}^n \hat{\Delta}_{i',j}\hat{\Delta}_{i',l}\right]
    \end{eqnarray*}
    We can see that the points $(i,j),(k,l),(i',j),(i',l)$ form a right trapezoid. Summing all those trapezoids up, for each specific $i,j,l,i'$,
    
    $$\frac{2}{A^2}\left(1-\frac{1}{A}\right)^2\hat{\Delta}_{i,j}\hat{\Delta}_{i',j}\hat{\Delta}_{i',l}\underbrace{\sum_{k\ne i}\hat{\Delta}_{k,l}}_{=-\hat{\Delta}_{i,l}}$$
    Summing all of them up gives us
    $$-\frac{4}{A^2}\left(1-\frac{1}{A}\right)^2\mathcal{R}$$
    \item If $\{i',k'\}=\{i,k\}$, then it's also similar to the previous case, adding $-4\mathcal{R}(1-1/A)^2/A^2$ to the variance.
    
    \end{itemize}
\end{enumerate}
It is quite hard to compute a value for $\mathcal{R}$, but we can compute an upper bound, for a rectangle $(i,j),(k,j),(k,l),(i,l)$,
$$\hat{\Delta}_{i,j}\hat{\Delta}_{k,j}\hat{\Delta}_{k,l}\hat{\Delta}_{i,l}\leq \frac{\hat{\Delta}_{i,j}^2\hat{\Delta}_{k,l}^2+\hat{\Delta}_{k,j}^2\hat{\Delta}_{i,l}^2}{2}$$
Therefore, an upper bound for the sum of all terms containing $\hat{\Delta}_{i,j}^2$ is going to sum up to $2\hat{\Delta}_{i,j}^2\sum_{k,l\in[n]}\hat{\Delta}_{k,l}^2\leq 2\hat{\Delta}_{i,j}^2\lVert \hat{\Delta}\rVert^2$, which, summing over all $i,j$ gives $2\lVert \hat{\Delta}\rVert^4$.

This completes the all parts for the variance and we can calculate that
\begin{eqnarray*}
\Var[X]&\leq&\left[\frac{4}{A}\left(1-\frac{1}{A}\right)^2+\frac{2}{A^2}\left(1-\frac{1}{A^2}\right)-\frac{8}{A^2}\left(1-\frac{1}{A}\right)^2\right]\mathcal{R}+\frac{2}{A^2}\left(1-\frac{1}{A^2}\right)\lVert\hat{\Delta}\rVert^4\\
&\leq&2\left[\frac{4}{A}\left(1-\frac{1}{A}\right)^2+\frac{2}{A^2}\left(1-\frac{1}{A^2}\right)-\frac{8}{A^2}\left(1-\frac{1}{A}\right)^2\right]\lVert\hat{\Delta}\rVert^4+\frac{2}{A^2}\left(1-\frac{1}{A^2}\right)\lVert\hat{\Delta}\rVert^4\\
&=&\left(\frac{8}{A}-\frac{26}{A^2}+\frac{40}{A^3}-\frac{22}{A^4}\right)\lVert\hat{\Delta}\rVert^4\\
&\leq&\frac{8}{A}\lVert\hat{\Delta}\rVert^4 \qquad (\text{when }A\geq 2)
\end{eqnarray*}
In particular, when $A=1$, we can see that $\Var[X]=0$, which is as expected since the output $\ell_2$ difference will always be 0. Therefore, we can conclude that
$$\Var[\lVert\Tilde{\Delta}\rVert^2]=\Var[\lVert\hat{\Delta}\rVert^2+X]=\Var[X]\leq\frac{8}{A}\lVert\hat{\Delta}\rVert^4$$
\newpage
\bibliography{reference}{}
\bibliographystyle{plain}
\end{document}